\documentclass{article}
\usepackage{times}
\usepackage{soul}
\usepackage{url}
\usepackage[hidelinks]{hyperref}
\usepackage[utf8]{inputenc}
\usepackage[small]{caption}
\usepackage{graphicx}
\usepackage{amsmath}
\usepackage{amsthm}
\usepackage{booktabs}
\usepackage{algorithm}
\usepackage{algorithmic}
\usepackage[switch]{lineno}

\usepackage{authblk}
\usepackage{fullpage}
\usepackage{enumerate}

\title{Learning Tree Automata with Term Rewriting}

\author{Jakub Kopystiański}
\author{Jan Otop}
\affil{University of Wrocław}
\date{}

\usepackage{amsmath, amsthm}
\usepackage{todonotes}
\presetkeys{todonotes}{inline,backgroundcolor=green}{}
\usepackage{amssymb}
\usepackage{xspace}

\newtheorem{theorem}{Theorem}
\newtheorem{lemma}[theorem]{Lemma}
\newtheorem{definition}[theorem]{Definition}
\newtheorem{fact}[theorem]{Fact}

\newtheorem{proposition}[theorem]{Proposition}
\newtheorem{remark}[theorem]{Remark}
\newtheorem{corollary}[theorem]{Corollary}

\newcommand{\N}{{N}}
\newcommand{\set}[1]{\{#1\}}

\newcommand{\arity}{\textsf{Ar}}
\newcommand{\dom}{\textrm{dom}}
\newcommand{\countAut}[2]{\#_{#1}[#2]}
\newcommand{\countAutCumm}[2]{\Sigma\#_{#1}[#2]}

\newcommand{\coNP}{\textsf{coNP}}
\newcommand{\NP}{\textsf{NP}}
\newcommand{\PTime}{\textsf{P}}

\newcommand{\rewritesOneStep}{\rightarrow}
\newcommand{\rewrites}{\rewritesOneStep^*}
\newcommand{\normalFormTRS}[2]{\textbf{NF}_{#1}(#2)}
\newcommand{\rew}{\mathcal{R}}

\newcommand{\subst}{\sigma}
\newcommand{\substExt}{\widehat{\subst}}

\newcommand{\Sig}{\mathcal{F}}
\newcommand{\trees}{\mathcal{T}}
\newcommand{\aut}{\mathcal{A}}
\newcommand{\autB}{\mathcal{B}}
\newcommand{\lang}{\mathcal{L}}

\newcommand{\Qfin}{Q_{\textrm{acc}}}

\newcommand{\deltaFinal}{\widehat{\delta}}
\newcommand{\deltaExt}[1]{\vec{\delta}_{#1}}

\newcommand{\congL}{\sim_{\lang}}

\newcommand{\lStar}{L\ensuremath{^*}}
\newcommand{\Vars}{\mathcal{X}}

\newcommand{\targetLang}{\mathcal{U}}

\newcommand{\subsumedBy}{\leq_{\aut}}
\newcommand{\subsumedByStrict}{<_{\aut}}
\newcommand{\subsumedByFunc}{\preccurlyeq_{\aut}}
\newcommand{\succBy}{\textsf{succ}_{\aut}}

\newcommand{\Iff}{\leftrightarrow}

\newcommand{\DLand}{\sqcap}
\newcommand{\DLor}{\sqcup}

\newcommand{\Lwords}{\mathcal{K}}
\newcommand{\treeLangFromWords}[1]{\mathcal{T}^{#1}}
\newcommand{\SigUnary}{\Sig_{\textrm{unary}}}

\begin{document}

\maketitle

\begin{abstract}
We present an extension of the Angluin-style learning algorithm for tree automata that incorporates deductive inference. 
The learning algorithm is provided with a term rewriting system that specifies properties of the target tree language (e.g., the order of subtrees under a symbol f is irrelevant). 
This term rewriting system is used to infer answers to some queries, which reduces the query complexity of the learning algorithm. 
We present examples of rewrite systems that express natural properties of tree-structured data, which yield a significant reduction in the number of queries.

 \end{abstract}

\section{Introduction}
\paragraph{Active automata learning.}
The aim of \emph{active automata learning} is to generate a \emph{deterministic finite automaton (DFA)} recognizing an unknown regular language 
$\lang$ (target language) having access to an oracle answering queries about $\lang$, i.e.,
the algorithm is allowed to \emph{actively} query about the language in contrast to \emph{passive} learning where
the input dataset does not change. 
There are two types of queries: 
\emph{membership} of a~given word in $\lang$, and 
\emph{equivalence} of the language of a given DFA $\aut$ with the target language $\lang$, where the answer to an equivalence query is either YES, 
or a \emph{counterexample} word distinguishing $\lang$ and $\lang(\aut)$.
This framework was proposed in the seminal paper~\cite{angluin1987learning} along with the \lStar-algorithm, which enables learning in polynomial time. 
Active automata learning has been intensively studied
for both efficiency~\cite{TTTalgorithm,Lsharp} as well as extensions to other automata-based models. 
In particular, the \lStar-algorithm has been adapted to tree automata~\cite{DBLP:conf/dlt/DrewesH03} and
multiplicity tree automata~\cite{DBLP:conf/icgi/HabrardO06,MarusicW15}.

\paragraph{Applications of active tree-automata learning.}
Trees are a natural model for structured data and hence tree automata are often used in knowledge representation, including processing of XML, JSON, and 
unification in Description Logics~\cite{BaaderTrees}.
Motivated by these applications, active learning of tree languages and related formalisms has been extensively studied and applied to 
information extraction~\cite{DBLP:conf/ijcai/KosalaBBB03}, 
the synthesis of node specifications~\cite{DBLP:conf/icdt/GrienenbergerR19}, 
structured data transformations~\cite{DBLP:conf/pldi/YaghmazadehKDC16}, 
and Description Logic $\mathcal{EL}$ terminologies~\cite{learningEL}.
Due to the close link between regular tree and context-free languages, active learning of tree languages has been applied to 
context-free grammar repair~\cite{DBLP:conf/sle/Raselimo021},
learning probabilistic grammars~\cite{DBLP:conf/aaai/NitayFZ21}, and
verification and explainability of RNNs~\cite{barbot2021extracting}.

\paragraph{Query complexity.}
The number of queries required by the \lStar-algorithm is polynomial in the size of the target automaton.
While it is feasible to answer that many queries automatically, the number remains too high for manual intervention.
Consequently, there is a large body of work on improving query complexity of 
active automata learning~\cite{TTTalgorithm,ADTalgorithm,Lsharp,KrugerJR24,DierlFHJST24} as well as implementing the teacher
using Large Language Models~\cite{learningEL,bhattamishraautomata,vazquezchanlatte2025ll}.
However, LLMs only answer membership queries, and their equivalence query answers are merely approximations with a large number of 
membership queries.
This approximation approach is similar to prior work on automating query answering~\cite{ModelLearning}.
Therefore, in this work, we follow an alternative route, proposed in~\cite{ecai25}, which is to restrict the search space of automata by
providing the learning algorithm with additional information (called advice) regarding the target language.  
Furthermore, we focus on inferring the answers to equivalence queries as 
it yields the greatest reduction in complexity and
increases reliability of learning (answers to these queries are often approximated).
We briefly discuss inference of answers to membership queries.

\paragraph{Learning with advice.}
We consider an extended learning framework, in which the learning algorithm is given \emph{advice} about the target language~\cite{ecai25}. 
This advice constrains the search space and allows the algorithm to infer answers to certain queries. 
This approach bridges two synthesis paradigms: 
deductive synthesis from specifications and inductive synthesis based on queries.
Furthermore, it allows for flexibility to express the learned language partially with queries and partially 
with a \emph{term rewriting system (TRS)}, which substantially differs from automata,
and hence some properties can be succinctly expressed with a TRS instead of
multiple queries.

\paragraph{Term rewriting as advice.}
A Term Rewriting System (TRS) consists of rewrite rules $l \rightarrow r$, where $l,r$ are terms (or equivalently trees with variables).
\emph{Rewriting} is the process of iteratively transforming a tree by replacing an instance of a left-hand side $l$ with its corresponding right-hand side $r$.
For example: 
\begin{itemize}
\item Commutativity: The rule $f(X,Y) \to f(Y,X)$ allows for the swapping of subtrees at node $f$.
\item Associativity: With  $f(f(X,Y),Z) \to f(X, f(Y,Z))$ any tree with $f$ symbols and constants can be transformed 
into a right skewed tree.
\end{itemize}
These rules serve as natural advice when the target language is known to satisfy such structural properties.
In our approach, an advice TRS defines an equivalence relation relative to the target language $\lang$.
Specifically, if a tree $t$ is in $\lang$, then all trees obtained by rewriting $t$ must be in~$\lang$. Similarly rewriting is compatible with the complement of $\lang$. 
We also consider one-sided relations stating that 
(1)~trees from $\lang$ are rewritten to trees from $\lang$  only, while there are no restrictions for trees not in $\lang$, or
(2)~trees not in $\lang$ are rewritten to trees not in $\lang$ only.
This allows the TRS to capture high-level structural properties, such as symmetry or associativity of a given symbol, while it is not expected to characterize the learned language completely. 

\paragraph{Contributions.}
Active automata learning with advice has been already introduced for finite automata on words~\cite{ecai25}. 
In this paper we study automata over trees and present the following substantial contributions:
\begin{enumerate}
    \item We develop an algorithm that infers answers to equivalence queries for tree automata (Section~\ref{s:advice}).
    The transition to trees introduces distinct technical challenges absent in the word case.
    In particular, general inference is computationally difficult. We therefore identify subclasses of TRSs that admit tractable inference.

    \item To illustrate the efficacy of the advice mechanism, we provide TRSs expressing natural tree language properties and show, through experimental evaluation, a significant reduction in query complexity (Section~\ref{sec:examples}).

    \item Finally, we introduce the synthesis problem for TRS for a given regular tree language, which is 
     the converse of active learning with advice: given a regular tree language, find a TRS consistent with it.
     We establish a connection between synthesis of TRS and 
     the classical problem of finding synchronizing words in DFA (Section~\ref{sec:synthesis}).   
\end{enumerate}
This is an extended version of the paper~\cite{ijcai26}. 
Whe code and experiments data are available at~\cite{experimentsRepo}.

\paragraph{Related work.}
There is a large body of work on optimizing the computational and query complexity of the \lStar-algorithm for word automata~\cite{TTTalgorithm,ADTalgorithm,Lsharp,DierlFHJST24} 
and tree automata~\cite{DBLP:journals/mst/DrewesH07,DBLP:journals/tcs/Kasprzik13}. 
These works are orthogonal to our approach, as our algorithm is compatible with any active learning algorithm using equivalence queries.
Our work directly extends~\cite{ecai25} to the tree setting. Independently, enhancing the performance of active tree automaton learning 
by restricting focus to specific subclasses of tree languages remains an active area of research~\cite{DBLP:conf/fossacs/HeerdtKR021,BjorklundBE17}. 
Furthermore, 
extending the types of queries to improve the efficiency of active learning has also been investigated~\cite{DBLP:journals/jalc/TirnaucaT07}.
 
\section{Preliminaries}
A \textbf{signature} (or ranked alphabet) $\Sig$ is a pair $(F,\arity)$ consisting of the set of symbols $F$ and a function 
$\arity \colon F \to \N$ assigning a unique \textbf{arity} to each symbol. 
To ease the notation, we will write $f \in \Sig$ meaning $f \in F$.

\subsection{Terms}

For a signature $\Sig$ and a set of \textbf{variables} $\Vars$, 
the set of \textbf{terms} over $\Sig$ and $\Vars$, denoted by $\trees(\Sig,\Vars)$,  is the least set containing $\Vars$
such that for all $f \in \Sig$, if $k = \arity(f)$ and $t_1, \ldots, t_k \in \trees(\Sig,\Vars)$, 
then $f(t_1, \ldots, t_k) \in \trees(\Sig,\Vars)$.
In particular, if $a \in \Sig$ and $\arity(a) = 0$, then $a \in \trees(\Sig,\Vars)$.
The set of \textbf{ground terms} $\trees(\Sig)$ over the signature $\Sig$ is the set $\trees(\Sig,\emptyset)$, i.e., 
terms containing no variables. 
We consider signatures $\Sig$ that  contain at least one constant symbol, as otherwise $\trees(\Sig)$ is empty.
A term $t$ is \textbf{linear} if every variable occurs at most once in $t$.

A \textbf{substitution} $\subst$ is a function from variables to terms with a finite domain. 
Each substitution $\subst$ can be uniquely extended to $\substExt$ over all terms  $\trees(\Sig,\Vars)$ as follows:
for $X \in \Vars$ we have
$\substExt(X) = \subst(X)$ if $X \in \dom(\subst)$ and otherwise $\substExt(X) = X$.
For a term $t$, which is not a variable, we have $t = f(t_1, \ldots, t_k)$, and 
$\substExt(t) = f(\substExt(t_1), \ldots, \substExt(t_k))$.
Since the extension from $\subst$ to $\substExt$ is unique, we refer to $\substExt$ as $\subst$. 
For a~variable $X$ and terms $t,s$, we denote by $t[X \to s]$ the term $\sigma_s(t)$ resulting 
from the substitution $\sigma_s(X) = s$. 

\subsection{Trees as terms}

A (ranked) {tree} over $\Sig$ is a ground term over $\Sig$ and we refer to $\trees(\Sig)$ as the set of trees over $\Sig$.
Referring to ranked trees as ground terms is common in the literature~\cite{tata}. 
In the unranked case, when symbols do not have unique arities, the order-based definition of trees is required, but 
in this paper we consider only ranked trees. 

A \textbf{context} $c$ is a term with a single occurrence of a variable $X$. 
For a context $c$ and a tree (ground term) $t$, we define $c(t) = c[X \to t]$, i.e., the substitution $\sigma(X) = t$ applied to $c$.

\subsection{Tree automata}

A bottom-up deterministic finite tree automaton (DFTA) is a tuple $(\Sig, Q, \Qfin, \delta)$
consisting of
    the signature $\Sig$, 
    a finite set of states $Q$, 
    a set of accepting states $\Qfin$, and
    a transition function  $\delta \colon Q^* \times \Sig \to Q$.
The transition function  $\delta$ is a partial function such that $\delta(u,f)$ is defined if and only if $|u|$ is equal to $\arity(f)$.   
We extend $\delta$ to a function $\deltaFinal \colon \trees(\Sig) \to Q$ inductively:
\begin{enumerate}
    \item If $t$ is a constant $c \in \Sig$, then $\deltaFinal(c)$ is the unique state $q$ such that $\delta(c) = q$.
    \item If $t = f(t_1, \ldots, t_k)$, then $\deltaFinal(t) = \delta(\deltaFinal(t_1), \ldots, \deltaFinal(t_k), f)$.
\end{enumerate}   
The tree language \textbf{recognized by} $\aut$, denoted by $\lang(\aut)$, is the set of all trees 
such that $\deltaFinal(t)$ is an accepting state, i.e., 
$\lang(\aut) = \set{t \mid \deltaFinal(t) \in \Qfin}$.

\subsection{Myhill-Nerode theorem for tree languages}
For a tree language $\lang$ over $\Sig$ we define $\congL$ on trees over $\Sig$ as follows:
for all trees $t_1, t_2$ over $\Sig$ we have $t_1 \congL t_2$ if and only if for every context $c$ over $\Sig$ 
we have $c(t_1) \in \lang \Iff c(t_2) \in \lang$. 

\begin{lemma}
A tree language $\lang$ is regular if and only if $\congL$ has finitely many equivalence classes. 
\end{lemma}

Furthermore, for every regular tree language $\lang$ a \textbf{minimal} DFTA recognizing $\lang$ exists 
and it is unique (up to an isomorphism). The states of the minimal DFTA for $\lang$ correspond 
to equivalence classes of $\congL$. Note that all states in the minimal DFTA are reachable.

\subsection{Learning tree automata}
\label{sec:angluin-trees}
The framework of active tree-automata learning assumes an oracle, called \emph{minimally adequate teacher}, 
which answers two types of queries about the target tree language $\lang$: 
\begin{itemize}
\item \textbf{membership queries}: given a tree $t \in \trees(\Sig)$, is $t \in \lang$?, and
\item \textbf{equivalence queries}: given a DFTA $\aut$, is $\lang(\aut) = \lang$? 
If not the teacher returns a counterexample, which is a tree from exactly one of the sets $\lang(\aut)$ and $\lang$.
\end{itemize}
The tree-automata variant of the \lStar-algorithm~\cite{DBLP:conf/dlt/DrewesH03} having access to
the oracle for a tree language $\lang$ returns the minimal DFTA $\aut_{\lang}$ recognizing ${\lang}$.
It works in polynomial time in $|\aut_{\lang}|$ and the total size of counterexamples supplied by the oracle. 
Unlike the word-automaton case, where the shortest counterexample is linear in the size of the minimal DFA,
a minimal-size counterexample tree can be of exponential size. 
However, if counterexamples are presented as directed acyclic graphs (DAGs),
their size remains polynomially bounded in $|\aut_{\lang}|$~\cite{charatonik1999automata,DBLP:journals/ipl/AnantharamanNR05}.

\subsection{Term rewriting systems}

A \textbf{term rewriting system (TRS)} $\rew$ over a signature $\Sig$ is 
a finite set of pairs of terms $(l,r)$ over $\Sig$.
A pair of terms  $(l,r)$  from $\rew$ is called a \textbf{rewrite rule} and denoted by $l \rightarrow r$.
For a TRS $\rew$, we define a \emph{single-step rewrite relation} $\rewritesOneStep_{\rew}$ over 
terms from $\trees(\Sig,\Vars)$ as the least relation such that
(i)~for all substitutions $\subst$ and all rules $(l,r)$ from $\rew$ we have $\subst(l) \rewritesOneStep_{\rew} \subst(r)$, and
(ii)~for all $f \in \Sig$, if $k$ is the arity of $f$ and $s \rewritesOneStep_{\rew} t$, 
then for all terms $s_1, \ldots, s_{k-1}$ and all $i$
we have $f(s_1, \ldots, s_{i-1},s,s_{i+1}, \ldots, s_k) \rewritesOneStep_{\rew} f(s_1, \ldots, s_{i-1},t,s_{i+1}, \ldots, s_k)$.
The rewriting relation $\rewrites_{\rew}$ is the transitive and reflexive closure of $\rewritesOneStep_{\rew}$.
We will omit the $\rew$ subscript if the TRS is clear from the context. 

\paragraph{Normal forms.}  
A term $t$ is in a \textbf{normal form} if there is no $t'$ such that $t \rewritesOneStep_{\rew} t'$. 
If for a term $s$ there is a unique $t$ in a normal form such that $s \rewrites_{\rew} t$, we say
that $t$ is the normal form of $s$ (w.r.t. $\rew$) and denote it by $\normalFormTRS{\rew}{s}$.

\paragraph{Computing normal forms.}  
A (finite) TRS $\rew$ is
\emph{terminating} if every sequence of terms $s_0, s_1, \ldots$ 
    such that $s_i \rewritesOneStep s_{i+1}$ is finite, 
\emph{confluent} if for all terms $s,s_1, s_2$, if $s \rewrites s_1$ and $s \rewrites s_2$, then there is $t$ such that
    $s_1 \rewrites t$ and $s_2 \rewrites t$, and
\emph{convergent} if it is terminating and confluent.
In a convergent TRS, every term $s$ has the (unique) normal form, 
which can be computed by applying reductions in an arbitrary order as long as the term can be reduced. 
Termination guarantees that this process takes finitely many steps and confluence entails the uniqueness of the result.
Therefore, for a convergent TRS $\rew$, one can effectively compute $\normalFormTRS{\rew}{s}$.
However, showing that a TRS is convergent is generally difficult~\cite{BaaderBook}.

\section{Learning tree automata with advice}
\label{s:advice}{
\newcommand{\lExt}{\hat{l}}
\newcommand{\rExt}{\hat{r}}
In this section, we adapt the framework of active learning with advice~\cite{ecai25} to tree automata. 
We first formally define the active learning problem with advice, then discuss how to resolve queries using that advice.

In our framework, advice for the learning  algorithm  is given via TRSs, which relate to tree languages 
through the following notions of consistency:

\begin{definition}[\cite{ecai25}]
\label{def:consistent}
Let $\Sig$ be a signature, $\rew$ be a TRS over $\Sig$ and $\lang$ be a (regular) tree language over $\Sig$. 
We say that 
\begin{itemize}
\item $\rew$ is \emph{(fully) consistent} with $\lang$ if and only if
for all trees $s,t$, if $s \rewrites_{\rew} t$, then $s \in \lang \iff t \in \lang$.  
\item $\rew$ is \emph{positively consistent} with $\lang$ if and only if for all trees $s,t$, if $s \rewrites_{\rew} t$, then $s \in \lang \implies t \in \lang$.  
\item $\rew$ is \emph{negatively consistent} with $\lang$ if and only if for all trees $s,t$, 
if $s \rewrites_{\rew} t$, then $s \notin \lang \implies t \notin \lang$.  
\end{itemize}
\end{definition}

Note that a TRS $\rew$ is fully consistent with $\lang$ if and only if it is both positively and negatively consistent.
While more granular than standard consistency, positive and negative consistency are harder to infer.

Having the notions of consistency, we can formally state the active learning with advice problem.

\begin{definition}
The active tree-automata learning with advice problem for an unknown regular tree language $\targetLang$ is as follows:
\begin{itemize}
\item \textbf{Input}: 
(1)~an oracle answering membership and equivalence queries about the target regular tree language $\targetLang$, and
(2)~three advice TRSs: 
$\rew_{=}$, which is consistent with $\targetLang$, and 
$\rew_{+}, \rew_{-}$, which are respectively positively and  negatively consistent with $\targetLang$.
\item \textbf{Output:} the minimal DFTA that recognizes $\targetLang$.
\end{itemize}
\end{definition}

\paragraph{Reducing equivalence queries with advice.}
To reduce the number of equivalence queries, 
the algorithm first checks whether an advice TRS $\rew$ is consistent (resp., positively or negatively consistent) with
the language of a candidate automaton $\aut$ before querying the oracle. 
If it is not consistent, an equivalence query is unnecessary because $\lang(\aut)$ must differ from the target language 
with which $\rew$ is known to be consistent (resp., positively or negatively consistent).
In such cases, the algorithm computes a pair of trees violating consistency (resp., positive or negative consistency). 
For example, for positive consistency, violating trees $s,t$ are such that $s \rewrites_{\rew} t$ 
but $s \in \lang(\aut) \implies t \in \lang(\aut)$ does not hold.
One of the trees $s$ or $t$ has to be the counterexample to $\lang(\aut)$  being equal to the target language. 
The algorithm decides which one with a single membership query.

\begin{remark}[Rewrite rules vs. regular languages]
Rewrite rules can express non-regular properties.
For example, $f(X,X) \to a$ expresses that the value of 
a tree with identical immediate subtrees the same as the value of $a$ with respect to $\lang$. 
Tree automata cannot express equality of subtrees. 
Similarly, 
while associativity $f(X,f(Y,Z)) \to f(f(X,Y),Z)$ does not compare subtrees, applying it iteratively can reshape trees, 
destroying information encoded in the tree structure. 
For instance, consider a regular tree language over constants $a,b$ and binary $f$ such that  
the string of leaf labels forms a palindrome and every palindrome appears in some tree. 
This language is defined to contain constants $a,b$ and be closed under the contexts
 $f(a,f(X,a))$ and $f(b,f(X,b))$. 
A DFTA can detect whether a tree is built from these two contexts and constants. 
However, if such tree is reshaped to a skewed tree, no DFTA can recognize whether the leaves form a palindrome.
Consequently, this language cannot be extended to satisfy associativity while retaining its key properties.
Despite this, we only ask whether a language \emph{known to be regular} satisfies rewrite rules. 
\end{remark}

\subsection{Deciding consistency}
\label{s:deciding-consistency}

We discuss how to efficiently decide consistency.
First, we establish a condition that reduces checking consistency of a TRS with a tree language 
to checking combinatorial properties of the minimal DFTA recognizing that language~(Lemma~\ref{l:compute-inconsistency-witnesses}).
Next, we discuss the complexity of deciding this combinatorial condition.    

To conveniently define the combinatorial condition, 
we extend the function $\deltaFinal$ from trees (i.e., ground terms) to all terms, 
where terms with variables define \emph{state transformations}:

\begin{definition}
Let $\aut$ be a minimal DFTA and $t$ be a term over variables $X_1, \ldots, X_k$. 
We define $\deltaExt{\aut}(t)$ as a function from $Q^k$ to $Q$ such that $\deltaExt{\aut}(t)(q_1, \ldots, q_k)$ 
is the state assigned to $t$ if in leaves labelled by $X_i$ the DFTA $\aut$ starts in state~$q_i$. 
Formally, let $s_1, \ldots, s_k$ be trees over $\Sig$ such that $\deltaFinal(s_i) = q_i$. 
Then, $\deltaExt{\aut}(t)(q_1, \ldots, q_k) = \deltaFinal(t[X_1 \to s_1, \ldots, X_k \to s_k])$.
\end{definition} 

Note that in a minimal DFTA every state is reachable, and hence trees $s_1, \ldots, s_k$ as in the above definition exist.

\begin{lemma}
\label{l:compute-inconsistency-witnesses}
Let $\rew$ be a TRS over $\Sig$ and $\aut$ be a minimal DFTA over $\Sig$. 
The TRS $\rew$ is \emph{consistent} with $\lang(\aut)$ if and only if 
the functions $\deltaExt{\aut}(l)$ and $\deltaExt{\aut}(r)$ coincide, 
for each rule $l \rightarrow r \in \rew$.
\end{lemma}
\begin{proof}
We show the implication from right to left.
Assume that $\rew$ is not {consistent} with $\lang(\aut)$. Then, there are trees $s,t$
such that $s \rewrites_{\rew} t$ and either $s \in \lang(\aut)$ and $t \notin \lang(\aut)$, or 
$s \notin \lang(\aut)$ and $t \in \lang(\aut)$.
Without loss of generality, assume that $s \in \lang(\aut)$ and $t \notin \lang(\aut)$.
Then, there is a rewriting sequence $s = t_0 \rightarrow t_1 \ldots \rightarrow t_n = t$.
Consider the first position $i$ such that $t_i \in \lang(\aut)$, but $t_{i+1} \notin \lang(\aut)$. 
Since $t_{i} \to t_{i+1}$, there is a context $c$ and the grounding substitution $\sigma$ such that
$t_{i} = c[\sigma(l)]$ and $t_{i+1} = c[\sigma(r)]$. 
Let $q_1 = \deltaFinal(\sigma(l))$ and  $q_2 = \deltaFinal(\sigma(r))$. These states have to be different as otherwise 
both $t_{i}$ and $t_{i+1}$ would have the same value in $\lang(\aut)$. 
It follows that on the states that correspond to the grounding substitution $\sigma$ 
the function $\deltaExt{\aut}(\sigma(l))$ and $\deltaExt{\aut}(\sigma(r))$ differ as one returns $q_1$ and the other $q_2$.

For the implication from left to right. Let $l \to r \in \rew$ be the rule such that 
the functions $\deltaExt{\aut}(l)$ and $\deltaExt{\aut}(r)$ are different. 
Consider
$\vec{q}$ that  witnesses that $\deltaExt{\aut}(l)$ and $\deltaExt{\aut}(r)$ are different.
Based on $\vec{q}$  we define the substitution $\sigma$ such that $\sigma(X_i)$ is the tree $s_i$ satisfying 
$\deltaFinal(s_i) = \vec{q}[i]$. Then,
$\deltaFinal(\sigma(l)) \neq \deltaFinal(\sigma(r))$. As $\aut$ is minimal, different states are distinguished with contexts and hence
there exists $c$ such that $c[\sigma(l)] \in \lang(\aut)$ and $c[\sigma(r)] \notin \lang(\aut)$ or vice versa. 
In any case, $l \to r \in \rew$ is not consistent with $\lang(\aut)$.
\end{proof}

It is essential that the TRS and the DFTA are over the same signature~$\Sig$, as otherwise checking consistency 
is undecidable~\cite{ecai25}.

\begin{remark}
Consider a rule $f(X,X) \to a$, for which the left-hand side defines a non-regular language.
Observe that Lemma~\ref{l:compute-inconsistency-witnesses} implies that if $f(X,X) \to a$ is consistent with the language of a minimal DFTA $\aut$, then
for all trees $s_1, s_2$ labelled with the same state by $\aut$ ($\deltaFinal(s_1) = \deltaFinal(s_2)$), the tree 
$f(s_1,s_2)$ is labeled with the state $\deltaFinal(a)$, which is a substantially stronger condition than just $f(X,X) \to a$.
\end{remark}

Now, in the remaining part of this section, we discuss the complexity of evaluating the condition from Lemma~\ref{l:compute-inconsistency-witnesses}.

For any tuple of states $\vec{q} \in Q^k$, 
the value $\deltaExt{\aut}(l)(\vec{q})$ (resp., $\deltaExt{\aut}(r)(\vec{q})$) 
can be computed in polynomial time in $|l|$ (resp., $|r|$). 
However, the number of arguments is equal to the number of variables in terms $l,r$ and 
hence there may be exponentially many tuples of states to check.  
Still, Lemma~\ref{l:compute-inconsistency-witnesses} implies that consistency can be checked in $\coNP$:
\begin{proposition}
Checking consistency of a TRS with a tree language given by a DFTA is in $\coNP$.
For every $N>0$, consistency is decidable in polynomial time over TRSs with at most $N$ variables, i.e., 
in every rule $l \to r$, both terms $l,r$ are over $\{X_1, \ldots, X_N\}$.
\end{proposition}

In particular, the consistency problem is decidable in polynomial time if the TRS is fixed and only the DFTA is the input.

\paragraph{Computing counterexamples.} 
Given a rule $l \rightarrow r \in \rew$ such that $\deltaExt{\aut}(l) \neq \deltaExt{\aut}(r)$ and 
$q_1, \ldots, q_k$ witnessing this, i.e., 
$\deltaExt{\aut}(l)(q_1, \ldots, q_k) \neq \deltaExt{\aut}(r)(q_1, \ldots, q_k)$, 
we can compute two trees witnessing consistency being violated. 
Let $s_1, \ldots, s_k$ be trees such that  $\deltaFinal(s_i) = q_i$.
Consider trees:
\begin{itemize}
    \item $\lExt = l[X_1 \to s_1, \ldots, X_k \to s_k]$, and 
    \item $\rExt = r[X_1 \to s_1, \ldots, X_k \to s_k]$. 
\end{itemize}     
Let $q_l = \deltaFinal(\lExt)$ and $q_r  = \deltaFinal(\rExt)$. Then, $q_l \leq q_r$.
Since $\aut$ is minimal, there is a context $c$ distinguishing states $q_l$ and $q_r$.
Thus, for $s = c(\lExt) $ and $t = c(\rExt)$, exactly one of 
$s,t$ belongs to $\lang(\aut)$.
In summary, while $s \rewrites t$ in one step, $s \in \lang(\aut) \Iff t \in \lang(\aut)$ does not hold.
Since we know that $s \in \targetLang \Iff t \in \targetLang$, one of $s,t$ is a counterexample.

In general, the $\coNP$ upper bound cannot be improved:

\begin{lemma}
Checking consistency of a rule $l \to r$, where $l$ is a (non-linear) term and $r$ is a ground term  
with the language of a given DFTA is $\coNP$-hard.
\label{l:consistency-np-hard}
\end{lemma}
\begin{proof}
We reduce the tautology problem for DNF formulas problem to our problem. 
Let $\Sig$ consists of binary $\land, \lor$, unary $\lnot$ and constants $\top, \bot$.
We define a DFTA $\aut_{\textrm{eval}}$ having two states $\{q_0,q_1\}$, which evaluates a given tree over $\Sig$ to $q_1$ if the corresponding 
logical expression is true and $q_0$ otherwise. The state $q_1$ is accepting. 
Now, the tree $\top$ is accepted by $\aut_{\textrm{eval}}$. 
We transform a given DNF formula $\varphi$ into the corresponding term $t_{\varphi}$ over 
$\Sig$. Then, the formula $\varphi$ is a tautology if and only if 
for every substitution $\sigma$ we have $\deltaFinal(\sigma(l)) = q_1$. 
The latter is equivalent to the rule $t_{\varphi} \to \top$ being consistent with the language of $\aut_{\textrm{eval}}$.
\end{proof}

The above lemma holds even for a fixed DFTA. 
Still, the hardness argument relies on variables occurring multiple times. 
Therefore, a natural question is what is the complexity of checking consistency of rules $l \to r$, in which both terms are linear.
We leave this as an \emph{open problem} and propose a simple heuristic.

\paragraph{Heuristics for linear rules.} 
We propose a heuristic for linear terms based on the \textbf{state-counting function}, 
which counts the cardinality of the preimage $\deltaExt{\aut}(l)$ and $\deltaExt{\aut}(r)$ with multiplicities. 
This heuristic does not confirm consistency, but it can identify some rules that violate consistency.
This is not a problem as consistency itself overapproximates equivalence. 
Consequently, any heuristic that overapproximates consistency also 
overapproximates equivalence, i.e., the heuristic does not err on the negative answers. 
 
\begin{definition}
For a term $t$ and a DFTA $\aut$, we define the state-counting function $\countAut{\aut}{t} \colon Q \to \N$ as 
the cardinality of the preimage $\deltaExt{\aut}^{-1}[\set{q}]$, i.e., 
$\countAut{\aut}{t}(q) = |\set{(q_1, \ldots, q_k) \mid \deltaExt{\aut}(t)(q_1, \ldots, q_k) = q }|$.
\end{definition}

Observe that if $\deltaExt{\aut}(l) = \deltaExt{\aut}(r)$, then $\countAut{\aut}{l} = \countAut{\aut}{r}$. 
Consequently, checking $\countAut{\aut}{l} = \countAut{\aut}{r}$ serves as a necessary (but not sufficient) condition for consistency of $l \to r$ with $\lang(\aut)$ (i.e., it is an underapproximation).  
This provides a computationally efficient heuristic to filter out inconsistent rules before performing more expensive checks.
Furthermore, $\countAut{\aut}{t}$ can be efficiently computed for linear terms $t$, and hence it can be used as a preliminary test even if consistency could be decidable in polynomial but superlinear time.

\begin{lemma}
The function $\countAut{\aut}{t}$ can be computed in time $O(|\aut| \cdot |t|)$ over linear terms $t$.
\end{lemma}
\begin{proof}
We simply compute the function bottom-up. For $t = f(t_1, \ldots, t_k)$,
having $\countAut{\aut}{t_1}, \ldots, \countAut{\aut}{t_k}$ it suffices to iterate over all tuples $(q_1, \ldots, q_k) \in Q^k$, 
where $k = \arity(f)$, and sum 
$\countAut{\aut}{t_1}[q_1] \cdot \countAut{\aut}{t_2}[q_2] \cdot \ldots \cdot  \countAut{\aut}{t_k}[q_k]$ over all tuples.
For each node in the tree $t$ the computation takes  $O(|\aut|)$  operations, which gives us our bound.
\end{proof}

Observe that if one of the functions $\countAut{\aut}{l}$ or $\countAut{\aut}{r}$ is constant, then $\countAut{\aut}{l} = \countAut{\aut}{r}$ implies 
$\deltaExt{\aut}(l) = \deltaExt{\aut}(r)$.  In particular, for rules $l \to r$ such that $l$ is a linear term and $r$ is a ground term, 
we can decide consistency in linear time.

\begin{corollary}
The consistency of a rule $l \to r$ with $\lang(\aut)$ 
can be performed in time  $O((|l|+|r|)\cdot |\aut|)$
for rules $l \to r$ such that
$l$ is a linear term and $r$ is a ground term.
\end{corollary}

\subsection{Deciding positive and negative consistency}
\label{s:deciding-pos-consistency}

The notions of positive and negative consistency are dual. Observe that $\rew$ is positively consistent with $\lang$ 
if and only if $\rew$ is negatively consistent with the complement of $\lang$.
Therefore, we will focus on positive consistency as the results carry over to negative consistency straightforwardly 
(note that DFTA can be complemented via swapping accepting and non-accepting states).

First, we present a counterpart of Lemma~\ref{l:compute-inconsistency-witnesses} for positive consistency. 
For that, we need to define two order relations: one on states and another on states transformations that rank states from the least accepting to the most accepting.

\begin{definition}
Let $\aut$ be a DFTA. We define an order relation $\subsumedBy$ on states of $\aut$ such that
$q_1 \subsumedBy q_2$ if for all trees $t_1,t_2$ such that 
$\deltaFinal(t_1) = q_1$ and $\deltaFinal(t_2) = q_2$, for every context $c$ we have $c(t_1) \in \lang(\aut) \implies c(t_2) \in \lang(\aut)$.
\end{definition}

Given a DFTA $\aut$, the relation $\subsumedBy$ can be computed in $O(|\aut|^3)$ by the greatest fixed-point iteration. 
The algorithm starts with the relation $R = Q \times Q$. It removes pairs $(q_1, q_2)$ such that 
$q_1 \in \Qfin$ and $q_2 \notin \Qfin$. 
Then, iteratively,
for every context $c$ of height $1$ (which corresponds to a transition), if $(c(q_1),c(q_2)) \notin R$, 
the pair $(q_1, q_2)$ is removed from $R$.
The greatest fixed-point is reached after $|Q|^2$ iterations, while each iteration can be performed in $O(|\aut|)$ 
by checking all contexts of height $1$. 

We extend $\subsumedBy$ from states to state transformations:

\begin{definition}
Let $\aut$ be a DFTA and $k>0$.
We define an order relation $\subsumedByFunc$ on functions from $Q^k$ to $Q$ as follows.
For all $h_1, h_2 \colon Q^k \to Q$ we have $h_1 \subsumedByFunc h_2$ if and only if 
$h_1(\vec{q}) \subsumedBy h_2(\vec{q})$, for all $\vec{q} \in Q^k$.
\end{definition} 

Finally, we present the counterpart of Lemma~\ref{l:compute-inconsistency-witnesses} for positive consistency:

\begin{lemma}
\label{l:compute-inconsistency-witnesses-positive}
Let $\rew$ be a TRS over $\Sig$ and $\aut$ be a minimal DFTA over $\Sig$. 
The TRS $\rew$ is \emph{positively consistent} with $\lang(\aut)$ if and only if 
$\deltaExt{\aut}(l) \subsumedByFunc \deltaExt{\aut}(r)$, for each rule $l \rightarrow r \in \rew$.
\end{lemma}
\begin{proof}
The proof is virtually the same as for Lemma~\ref{l:compute-inconsistency-witnesses}. 
Assume that $\rew$ is not positively {consistent} with $\lang(\aut)$. 
Then, as in Lemma~\ref{l:compute-inconsistency-witnesses} there are $s,t$ such that 
$s \in \lang(\aut)$,  $t \notin \lang(\aut)$, and
$s \to t$ in one step. 
Therefore, there is a context $c$ and the grounding substitution $\sigma$ such that
$t_{i} = c[\sigma(l)]$ and $t_{i+1} = c[\sigma(r)]$. 
Let $q_1 = \deltaFinal(\sigma(l))$ and  $q_2 = \deltaFinal(\sigma(r))$. 
It follows that on the states that correspond to the grounding substitution $\sigma$ 
the relation $\deltaExt{\aut}(\sigma(l)) \subsumedByFunc  \deltaExt{\aut}(\sigma(r))$ 
does not hold.

Conversely, let $l \to r \in \rew$ be the rule such that 
the functions $\deltaExt{\aut}(l) \subsumedByFunc \deltaExt{\aut}(r)$ does not hold.  
Consider
$\vec{q}$ such that $\deltaExt{\aut}(l)(\vec{q}) \subsumedBy \deltaExt{\aut}(r)(\vec{q})$ does not hold.
Based on $\vec{q}$  we define the substitution $\sigma$ such that $\sigma(X_i)$ is the tree $s_i$ satisfying 
$\deltaFinal(s_i) = \vec{q}[i]$. Then,
$\deltaFinal(\sigma(l)) \subsumedBy \deltaFinal(\sigma(r))$ does not hold. 
Therefore, there is a context $c$ such that 
such that $c[\sigma(l)] \in \lang(\aut)$ and $c[\sigma(r)] \notin \lang(\aut)$. 
It follows that $l \to r \in \rew$ is not positively consistent with $\lang(\aut)$.
\end{proof}

Since $\subsumedBy$ can be computed in polynomial time, the condition $\deltaExt{\aut}(l) \subsumedByFunc \deltaExt{\aut}(r)$ can be falsified 
by finding an appropriate $\vec{q}$ and computing $q_1 = \deltaExt{\aut}(l)(\vec{q}), q_2 = \deltaExt{\aut}(r)(\vec{q})$ and finally checking 
that $q_1 \subsumedBy q_2$  does not hold. In consequence, checking positive consistency is in $\coNP$.

Observe that for rules $l \to r$ with a ground term $r \in \lang$, positive consistency and consistency coincide, i.e.,  
$l \to r$ is consistent with $\lang$ if and only if it is positively consistent. 
If $r \notin \lang$, we can instead take the complement of $\lang$ and reduce to the previous case. 
Therefore, for such rules the complexity of checking positive consistency is the same as checking (full) consistency: 
if $r$ is a ground term, checking consistency of $l \to r$ with a regular tree language $\lang$ 
is $\coNP$-complete in general, and it is decidable in polynomial time for rules $l \to r$ with $l$ being a linear term.
When $l,r$ are both linear terms, the complexity of checking consistency of $l \to r$ is left open. 
For positive consistency, we can show that it is $\coNP$-complete. It suffices to show hardness:

\begin{lemma}
Checking positive consistency of a rule $l \to r$ with the language of a given DFTA is $\coNP$-hard for rules $l \to r$ such that 
$l,r$ are linear terms.
\end{lemma}
\begin{proof}
\newcommand{\allEq}{\textsf{all-eq}}
As in the proof of Lemma~\ref{l:consistency-np-hard}, we show a reduction from the tautology problem.
We use a slight modification of the tautology problem. 
We say that a formula is in 3-CNF-3  if it is in a CNF, all clauses have at most $3$ literals and each variable has at most $3$ occurrences.
Analogously, a formula is in  3-DNF-3, if it is in a DNF, all disjuncts have at most $3$ literals and each variable has at most $3$ occurrences.
The negation of a 3-CNF-3 formula is a 3-DNF-3 formula.
The satisfiability problem for 3-CNF-3 formulas is  \NP-complete~\cite{DBLP:journals/dam/Tovey84}.
It follows that its complement, the tautology problem over 3-DNF-3 formulas is \coNP-complete.
Our reduction is from the tautology problem over 3-DNF-3 formulas.

Let $\Sig$ consists of a ternary $\allEq$, binary $\land, \lor$, unary $\lnot$ and constants $\top, \bot$.
Let $\aut_{\textrm{eval}}$ be a DFTA that evaluates logical expressions, where $\allEq$ is evaluated to true all of its 3 arguments have the same truth value.
Given a propositional formula  $\varphi$ in 3-DNF-3, 
we  translate it to $\psi$, in which variable in occurrences are distinct, i.e., 
a variable $x$ that occurs three times is substituted with $x_1, x_2, x_3$ in subsequent occurrences.
Then, let $t_{\psi}$ be a term over $\Sig$ corresponding to $\psi$.

Let $X$ be the set of variables from $\varphi$ and $Y$ be the set of variables from $\psi$. 
We build a formula $\xi$ over $\Sig$, which is true if all variables from $Y$, 
which correspond to the same variable from $X$ have the same logical value.
This formula is a conjunction of expression $\allEq(x_1, x_2, x_3)$, where every variable from $Y$ occurs exactly once.
Finally, observe that $t_{\xi} \to t_{\psi}$ is positively consistent with $\lang(\aut_{\textrm{eval}})$ if and only if 
$t_{\xi} \in \lang(\aut_{\textrm{eval}})$ implies $\psi \in \lang(\aut_{\textrm{eval}})$, if for every substitution $\sigma$.
This, in turn, holds exactly when for every variable assignment on $Y$, 
if this assignment comes from an assignment on $X$, then $\psi$ is true.
The latter holds if and only if $\varphi$ is a tautology.
\end{proof}

Still, we can utilize the heuristic based on $\countAut{\aut}{t}$, which was proposed above for consistency. 
We need to adapt the condition to positive consistency in the following way. 
We define the cumulative state-counting function $\countAutCumm{\aut}{t}$ as 
$\countAutCumm{\aut}{t}(q) = \sum_{q \subsumedBy q'} \countAut{\aut}{t}(q')$, i.e., it sums the multiplicity of 
each state $q'$ that is greater in the $\subsumedBy$ order.    

\begin{lemma}
For terms $l,r$, if $\deltaExt{\aut}(l) \subsumedByFunc \deltaExt{\aut}(r)$, then 
 $\countAut{\aut}{l}(q) \leq \countAutCumm{\aut}{r}(q)$, for every state $q$.
\end{lemma}

Note that having the function $\countAut{\aut}{r}(q_2)$ and the relation $\subsumedBy$, we can compute  $\countAutCumm{\aut}{t}(q)$ in time $O(|Q|^2)$ from definition or
in $O(|\subsumedByStrict|\cdot|Q|)$, where $\subsumedByStrict$ is the strict variant of the order $\subsumedBy$ and
$|\subsumedByStrict|$ is the number of tuples in $\subsumedByStrict$. 
The latter complexity is achieved by computing the \emph{transitive reduction} $\succBy$~\cite{transitiveReduction} of 
$\subsumedByStrict$ in $O(|\subsumedByStrict|\cdot|Q|)$, i.e., the transitive closure of $\succBy$ is $\subsumedByStrict$. 
Then the graph $(Q,\succBy)$ is sorted topologically in $O(|\succBy|\cdot|Q|)$ and  
the summation is performed iteratively over the immediate successors only in $O(|\succBy|\cdot|Q|)$. 
Since $\succBy \subseteq \subsumedByStrict$, the whole computation is in $O(|\subsumedByStrict|\cdot|Q|)$.

\subsection{Membership queries}
\label{s:membership-queries}
\newcommand{\rewMem}{\rew_{\textrm{mem}}}

Typical implementations of the \lStar-algorithm store answers to membership queries in a \emph{cache} to ensure that each query is unique. 
We extend this idea to membership queries modulo term rewriting.

First, for the advice TRS  $\rew_{=}$ (which is fully consistent with the target language), we select a convergent subset $\rewMem$.
Since the TRS $\rewMem$ is convergent, every term has a unique normal form. 

The algorithm maintains a {cache}, which is a dictionary mapping each tree's normal form $\normalFormTRS{\rewMem}{t}$ to its membership result. 
Before asking a query for a tree $s$, the algorithm computes the normal form of $s$ and checks the cache.  
If present, it fetches the stored answer without querying the oracle.
Consistency of $\rew$ with the target language $\lang$ implies that 
if  $\normalFormTRS{\rewMem}{s} = \normalFormTRS{\rewMem}{t}$, then $s \in \lang \Iff t \in \lang$.

Otherwise, if $\normalFormTRS{\rewMem}{s}$ is not in the cache, it asks the oracle the membership query and stores $\normalFormTRS{\rewMem}{s}$ 
along with the answer in the cache.

The TRSs $\rew_{+}, \rew_{-}$ can be used to infer answers to membership queries as well, but the complexity of inference makes this infeasible~\cite{ecai25}.

\section{Examples of advice rewrite rules}
\label{sec:examples}
We present examples of TRS that express interesting properties of tree languages.
While these properties originate in algebra, we discuss their relevance to knowledge representation, structured data, 
and Description Logic (DL)~\cite{DLbook}.  

\subsection{Families of TRS}

We consider the following types of TRSs: \emph{associativity},  \emph{commutativity},  
variants of \emph{distributivity}, and \emph{cancellation} rules. The applications of these rules are discussed in the context of
full consistency. We elaborate on these TRSs below.

\paragraph{Associativity.} 
While we have considered ranked trees, unranked trees with symbols of variable arities are prevalent in structured data.
For example, the HTML \texttt{<body>} tag has an arbitrary number of immediate successors, and 
dictionaries in JSON have arbitrarily 
many entries.
Unranked symbols can be encoded with associative binary symbols; the associativity rule for $f$: $f(X,f(Y,Z)) \to f(f(X,Y),Z)$ 
implies that the shape of a tree labelled with $f$ is irrelevant and hence all $f$ symbols can be contracted into 
a single high-arity symbol.
Furthermore, the basic DL operators $\DLand$ and $\DLor$ are associative.

\paragraph{Commutativity.}	
Commutativity of $f$ is expressed with the rule $f(X,Y) \to f(Y,X)$.
This can be more precisely described as \emph{horizontal commutativity}, indicating that  
sibling subtrees can be swapped without changing the tree's value.
JSON dictionaries and DL operators $\DLand$ and $\DLor$ are (horizontally) commutative. 
Furthermore, unranked and unordered trees~\cite{DBLP:conf/rta/BonevaT05} can be modeled with binary operators that are associative and (horizontally) commutative.

\paragraph{Distributivity.}
The standard distributive properties for binary symbols $f,g$ are:
(1)~$f$ is \emph{left-distributive} over $g$: $f(X,g(Y,Z)) \to g(f(X,Y), f(X,Z))$, and 
(2)~$f$ is \emph{right-distributive} over $g$: $f(g(X,Y),Z) \to g(f(X,Z), f(Y,Z))$.
DL operators $\DLand$ and $\DLor$ are left- and right-distributive one over another, i.e., both $\DLand$ over $\DLor$ and 
$\DLor$ over $\DLand$.

\paragraph{Variants of distributivity.}
Distributivity can be generalized beyond binary symbols. 
For a binary symbol $f$ and an unary symbol $g$ the rule $g(f(X,Y)) \to f(g(X),g(Y))$ is a variant of  distributivity.
Similarly, for unary $f,g$ the rule $g(f(X)) \to f(g(X))$ can be considered as a variant of distributivity. 
This rule states that the relative order of $f$ and $g$ along any root-to-leaf path is irrelevant, i.e., 
$f$ and $g$ commute along any path.
Consequently, variants of distributivity can be regarded as \emph{vertical commutativity}.
Typical examples of vertical commutativity are independent operations such as bold and italic text tags in HTML, which can be applied in any order, or more generally
an operator $g$ that is applied to all leaves of the subtree and hence $g(f(X,Y)) \to f(g(X),g(Y))$ holds.

\paragraph{Context cancellation rules.}
The general form of a \emph{context cancellation rule} is $c(t) \to c(X)$, where $c$ is a context and $t$ is a term.
The idempotency property of $f$ expressed with $f(f(X)) \to f(X)$ is a context cancellation rule. 
Similarly, there are various cancellation rules in DL: $X \DLand X \to X, X \DLor X \to X$ and $(X \DLand Y) \DLor X \to X$.

\section{Experiments}
In this section we discuss our experimental setup and the obtained results. 
We have considered associativity, distributivity and commutativity advice with the full consistency notion.
First, we discuss multiple learning settings that we have implemented and the research questions that have been addressed in our study.
Then, we discuss each advice type separately. For associativity and 
distributivity we discuss generation of datasets and obtained results, while for commutativity we briefly explain why there has been no improvement with our approach.

\subsection{Experimental Setup}
We have implemented the whole setup for learning tree languages in C\texttt{++}. 
It consists of the \lStar algorithm adapted to bottom-up DFTA and the oracle answering membership and equivalence queries based on a given DFTA.
The \lStar algorithm has been implemented in several variants: the classical one and $3$ variants with advice.
Our implementation accepts any set of rewrite rules as input, though our evaluation focuses on the TRS classes mentioned above.
We have implemented two algorithms for answering equivalences queries: the exact algorithm based on reachability in the product automaton, and
the approximate algorithm based on checking conformance of a given DFTA on $N$ randomly generated trees.
The value $N$ is a parameter of the algorithm, which has been estimated experimentally to achieve high probability of correct answers on automata used in the evaluation.
The cost of a single equivalence query is measured in \emph{tokens}, where a token is a node of a tree, and the number of tokens refers to the total number of nodes in trees used in the approximation of equivalence.
We have studied the approximate variant to simulate a scenario, in which equivalence queries cannot be directly implemented~\cite{ModelLearning,learningEL,bhattamishraautomata,vazquezchanlatte2025ll}. 
For instance, if the oracle is implemented based on running a computer program or querying an LLM, the  answers to equivalence queries are approximated with a large number of 
membership queries.
In summary, we have implemented the following learning settings:
\begin{enumerate}[(S1)]
\item The standard $\lStar$ algorithm for learning tree automata (our baseline) with the exact algorithm answering equivalence queries.
\label{stdSetting} 
\item
The \lStar algorithm with advice, in which the learning algorithm can use an advice term rewriting system to infer answers to equivalence queries. 
The consistency checking algorithm is the exact algorithm presented in Section~\ref{s:deciding-consistency}, which is exponential in the number of variables in the TRS.
\label{adviceSetting}
\item The approximate \lStar algorithm, in which the oracle implements the approximate algorithm for answering equivalence queries based on testing conformance of both automata on random trees. The learning algorithm is
the standard $\lStar$ algorithm.
\label{approxSetting}
\item The approximate \lStar algorithm with advice. The setting as in (S\ref{approxSetting}), but the learning algorithm can use an advice term rewriting system, with the exact algorithm for checking consistency (Section~\ref{s:deciding-consistency}).
\label{approxSettingAdvice}
\item The approximate \lStar with advice and using random testing in checking consistency with a TRS. 
This is as (S\ref{approxSettingAdvice}) except that the consistency checks are approximated with a random test 
similarly to approximate equivalence queries. 
\label{approxSettingAdviceApprox} 
\item  The approximate \lStar with advice and using counting heuristic in checking consistency with TRS. 
The setting as in (S\ref{approxSettingAdvice}), but before searching for counterexamples, the consistency algorithm runs  
the state-counting-based heuristic (Section~\ref{s:deciding-consistency}). Only if the heuristic confirms that the language of the candidate automaton is not consistent, the exact algorithm is executed.
\label{approxSettingCounting}
\end{enumerate}
 
In the experiments we have addressed the following research questions regarding query complexity, runtime complexity, accuracy and dependence on the cost of membership queries.

\begin{enumerate}[(Q1)]
\item (Query complexity) What is the impact of the advice mechanism on the number of equivalence queries posed to the oracle? 
\item (Time complexity) What is the impact of the advice mechanism on the overall learning runtime including runtime of the oracle.  
\item (Accuracy) In the approximate settings (S\ref{approxSetting}) --- (S\ref{approxSettingCounting}) the learning algorithm can return a DFTA that does not recognize the target language due to the approximate nature of equivalence tests.
The \emph{accuracy} refers to the ratio of correct DFTA returned by the learner.
What is the impact of the advice mechanism on accuracy of learning in the approximate settings?
\item (Break-even point) If the oracle has direct access to a DFTA, the exact equivalence test is more efficient than the approximate one. 
However, in many applications, the oracle has no direct access to the DFTA. Moreover, the membership queries are in fact more expensive than just evaluating a tree with respect to a given DFTA. 
Therefore, we study how the cost of membership queries influences the runtime of the learning algorithm without and with advice? 
What is the minimal cost of membership queries under which using advice is more efficient than the algorithm with no advice?
\end{enumerate}

In the following, we address the above questions for associativity and distributivity advice.

\subsection{Associativity}

We have evaluated the impact of the associativity advice  
$f(X,f(Y,Z)) \to f(f(X,Y),Z)$ 
on the learning process; we have addressed all research questions on randomly generated automata. 
To ensure that the language of a  randomly generated DFTA satisfies associativity, we generate the DFTA as follows. 

\paragraph{Generating DFTA satisfying associativity.}
Let $\Sig$ be a signature consisting of a binary symbol $f$ and constant symbols $\Sigma$.
The \emph{yield} of a tree $t$ over $\Sig$ is the sequence of constants $\Sigma$ in leaves read from left to right.
For a regular language $\Lwords$ over $\Sigma$, 
we define $\treeLangFromWords{\Lwords}$ as the tree language consisting of all trees with the yield from $\Lwords$. 
The membership of a tree in $\treeLangFromWords{\Lwords}$ does not depend on its shape, 
which can be expressed by stating that $f$ is associative. Intuitively, a tree $t$ over $\Sig$ can be considered as 
a tree of height one with a single variable-arity symbol $\mathbf{f}$ in the head and constants from $\Sigma$ as arguments, which
is effectively an encoding of a sequence over constant symbols $\Sigma$.
We show that tree languages satisfying associativity over $\Sig$ are exactly tree language obtained from regular languages via $\Lwords \mapsto \treeLangFromWords{\Lwords}$:
 
\begin{lemma}
\label{l:associativity}
Let  $\lang$ be a regular tree language over a binary symbol $f$ and constants $\Sigma$.
The rule  $f(X,f(Y,Z)) \to f(f(X,Y),Z)$ is consistent with $\lang$ if and only if
$\lang = \treeLangFromWords{\Lwords}$ for some regular language $\Lwords$ over $\Sigma$.
\end{lemma}
 \begin{proof}   
 \newcommand{\autD}{\mathcal{D}}
Observe that if $\lang = \treeLangFromWords{\Lwords}$ for some regular language $\Lwords$ over $\Sigma$
then it is consistent with  $f(X,f(Y,Z)) \to f(f(X,Y),Z)$. Indeed, 
for any trees $t_1, t_2, t_3$ the trees $f(t_1, f(t_2, t_3))$ and $f(f(t_1, t_2), t_3)$ have the same yield
and hence applying $f(X,f(Y,Z)) \to f(f(X,Y),Z)$ does not change the membership (or not) to $\treeLangFromWords{\Lwords}$.

Now, assume that the rule  $f(X,f(Y,Z)) \to f(f(X,Y),Z)$ is consistent with $\lang$ over $\Sig$.
We can show by induction on the size of the tree, that for any two trees $t_1, t_2$ with the same yield, 
there  exists a tree $t_3$ with the same yield as $t_1, t_2$ such that
$t_1 \rewrites t_3$ and $t_2 \rewrites t_3$. That tree $t_3$ satisfies the property that for every node labelled with $f$, the right subtree does not have $f$ in the root.
Therefore, for every word $w \in \Sigma^*$, the language
$\lang$ either contains all trees with the yield $w$ or none of such trees. 

In general, the word language of yields of trees from a given regular tree language is context-free, and 
any context-free language can be presented as a set of yields of some regular tree language. 
However, for a regular tree language satisfying associativity, the language of yields has to be regular. 

To see this, observe that if the language $\lang$ contains a tree with a yield $w \in \Sigma^*$, 
then it also contains a left skewed tree (exactly as $t_3$ above)
\[ f(f(\ldots f(f(w[1],w[2]),w[3])\ldots), w[n]) \]
where $n = |w|$.
We can intersect $\lang$ with the regular tree language of left skewed trees and obtain a regular tree language $\lang^{LS}$ 
with the same set of yields as $\lang$.
Then, on left skewed trees, a (bottom-up) DFTA works as a DFA that processes the word from left to right. 
More precisely, let $\aut_{LS}$ be a DFTA recognizing the tree language $\lang^{LS}$.
We can transform $\aut_{LS}$ to a DFA recognizing the word language of yields of $\lang^{LS}$.
Let $q_{1}, \ldots, q_{k}$ be states assigned by the DFTA $\aut_{LS}$ to constants $a_1, \ldots, a_k $ from $\Sigma$.
We define a DFA $\autD$ whose set of states consists of all states of $\aut_{LS}$ and a fresh initial state $q_0$, which is not a state of $\aut_{LS}$.
For every constant $a_i \in \Sigma$, we define $\delta_{\autD}(q_0,a_i) = q_i$. 
Next, for all states $s \neq q_0$ and constants $a_i \in \Sigma$, we define the transition of $\autD$ as $\delta_{\autD}(s,a_i) = \delta_{LS}(s,q_{i},f)$, where $\delta_{LS}$ is the transition function of $\aut_{LS}$. 
Accepting state of $\aut_{LS}$ and $\autD$ are the same.
Observe that the DFA $\autD$ recognizes the language of all yields of $\lang^{LS}$, which is equal to the language of all yields of $\lang$.
\end{proof}

\paragraph{The dataset.}
Based on Lemma~\ref{l:associativity},
 we have generated $935$ random DFA over the alphabets with $2$--$5$ letters, 
 with $2$--$16$ states, and converted each of them to a DFTA. 
These DFTA are guaranteed to satisfy associativity and their number of states ranged between $20$ and $200$. 
While the conversion of DFTA to DFA is linear, the conversion of DFA to DFTA incurs exponential blow-up as a DFTA has to process leaves in various orders depending on the shape of the tree (e.g., left-leaning vs. right-leaning).
For instance, for $w = aabb$, a DFTA, which works bottom-up on the tree $f(f(f(a,a),b),b)$, processes leaves left-to-right, 
while a DFTA on the tree $f(a,f(a,f(b,b)))$ processes leaves right-to-left. 
We have removed $94$ trivial cases in which the $\lStar$ algorithm finished learning with at most $2$ equivalence queries, leaving $841$ automata.

\paragraph{Results in the exact setting.}
We have compared settings (S\ref{stdSetting}) and (S\ref{adviceSetting}) regarding query and time complexity.
We have observed reduction of equivalence queries between 
$29\%$ (from $7$ to $5$ queries) to $96\%$ (from $49$
to $2$ queries) with an average reduction of $78.5\%$.
The median reduction is $83\%$.
In the absolute terms, this corresponds to a reduction from $12$ to $2.6$  equivalence queries on average.
The average runtime has increased slightly from $6.9$ seconds without advice to $7.02$ seconds with advice.
The complete notebook with the analysis of results is available at~\cite{experimentsRepo}.

\paragraph{Results in the approximate setting.}
We have compared the setting (S\ref{approxSetting}) without advice with settings (S\ref{approxSettingAdvice}) --- (S\ref{approxSettingCounting}) with advice.
We have observed the average time of learning with approximate equivalence queries (S\ref{approxSetting}) to be $6.22$ seconds, 
while it is $8.18$ seconds on average with the associativity advice (S\ref{approxSettingAdvice}) and $7.55$ seconds, if the state-counting heuristic (S\ref{approxSettingCounting}) is involved. 
The accuracy is respectively 
$91.3\%$ for (S\ref{approxSetting}),  
$96.7\%$ for (S\ref{approxSettingAdvice}), and 
$96.6\%$ for (S\ref{approxSettingCounting}). 
Thus, the accuracy gain is significant.
We have also considered (S\ref{approxSettingAdviceApprox}), in  which the consistency checks are approximated as well. 
We have observed a significantly worse time of $43.18$ seconds on average, but the accuracy $95.4\%$ is still better than without advice.
Note that unlike in the equivalence queries, approximation of the consistency checks as in (S\ref{approxSettingAdviceApprox}) is not necessary, 
since  the learning algorithm has direct access to a candidate automaton and hence it can run the exact consistency check.
Surprisingly, while using advice reduced the number of equivalence queries, the number of used tokens, i.e., the total size of the trees used for membership queries to approximate equivalence queries, has increased.  
We plan to research this phenomenon. 
The complete notebook with the analysis of results is available at~\cite{experimentsRepo}.

\subsection{Distributivity}

We have evaluated the impact of a variant of distributivity $g(f(X,Y)) \to f(g(X), g(Y))$ on the learning process; we have addressed all the research questions Q1 --- Q4 on randomly generated automata. 
We first discuss the generation process, which ensures that the language of a generated DFTA satisfies the rule  $g(f(X,Y)) \to f(g(X), g(Y))$.

\paragraph{The dataset.}
We have generated random DFTA over signatures consisting of a binary symbol $f$, a unary symbol $g$ and constants. 
Next, each random DFTA $\aut$ was forced to be consistent with the rule $g(f(X,Y)) \to f(g(X), g(Y))$ in the following way: 
for all states $q_1, q_2$, if there are $s_1, s_2$ such that
$\delta_{\aut}(s_1, g) = q_1$ and 
$\delta_{\aut}(s_2, g) = q_2$, then  $\delta_{\aut}(q_1, q_2, f)$ is set to $\deltaExt{\aut}(g(f(X_1,X_2)))(q_1, q_2)$. 
Since $f$ is also on the right hand side, this does not guarantee that the resulting automaton recognises the language consistent with $g(f(X,Y)) \to f(g(X), g(Y))$.
Therefore, we have finally checked consistency with $g(f(X,Y)) \to f(g(X), g(Y))$ and rejected automata whose language was not consistent with  $g(f(X,Y)) \to f(g(X), g(Y))$.
We have generated $655$ random DFTA and in $567$  cases, the resulting DFTA were consistent with the distributivity rule.
The random DFTA had between $5$ and $256$ states. 

\paragraph{Results in the exact setting.}
We have compared settings (S\ref{stdSetting}) and (S\ref{adviceSetting}) regarding query and time complexity.
We have observed reduction of equivalence queries 
between $0\%$ to $98\%$ (from $57$ to a single equivalence query) with an average reduction of $64\%$.
On average, this corresponds to a reduction from $13.4$ to $4.85$ equivalence queries.
The runtime has decreased from $8.19$ seconds on average to $5.09$ seconds with advice.
The complete notebook with results is available at~\cite{experimentsRepo}.

\paragraph{Results in the approximate setting.}
We have compared the setting (S\ref{approxSetting}) without advice, and settings (S\ref{approxSettingAdvice}) --- (S\ref{approxSettingCounting}) with advice.
We have observed the average time of learning with approximate equivalence queries (S\ref{approxSetting}) to be $6.65$ seconds.
For the settings with advice, the average times were:
$7.94$ seconds with the distributivity advice (S\ref{approxSettingAdvice}),
$7.17$ seconds with the distributivity advice and approximation of consistency checks (S\ref{approxSettingAdviceApprox}), and
$8.13$ seconds with the distributivity advice and the state-counting heuristic (S\ref{approxSettingCounting}). 
The accuracy is respectively
$90.3\%$ for (S\ref{approxSetting}),  
$91.2\%$ for (S\ref{approxSettingAdvice}), 
$88.9\%$ for (S\ref{approxSettingAdviceApprox}), and 
$90.8\%$ for (S\ref{approxSettingCounting}).
Thus, the accuracy does not increase significantly, and with approximation of consistency checks it can even decrease.
Finally, we have estimated that 
the break-even point is reached when the average cost of a membership query is at least $0.534$ ms per token.
The complete notebook with the analysis of results is available at~\cite{experimentsRepo}.

\subsection{Commutativity}

We also evaluated the commutativity advice: $f(X,Y) \to f(Y,X)$.
However, it did not yield any reduction in the number of equivalence queries.
To see the reason, observe that in a tree language consistent with commutativity, the tree $f(t_1,t_2)$ belongs to the language if and only if $f(t_2,t_1)$ does.
Therefore, in any candidate automaton computed based on membership queries, the transition function for $f$ is commutative. 
This shows some limitations of rewriting-based advice. 
We have observed a similar phenomenon in the word setting, where the idempotence string rewrite rule $aa \to a$ 
has yielded little reduction in the number of equivalence queries~\cite{ecai25}. 
This was also due to the fact that the \lStar algorithm primarily generates DFA that have self-loops unless there is a test word implying that a self-loop is impossible.
The idempotence rule $aa \to a$ does not imply that all states have self-loops over $a$, but all states reachable over $a$ need to have a self-loop over $a$ (in the minimal DFA).
Thus, the idempotence rule can improve query complexity, but the gain is minimal.

In both cases, the learning algorithm discovers the properties expressed by advice through membership queries rendering the advice redundant.	
This suggests that very short rewrite rules are unlikely to be effective as advice.

\section{Generating advice}
\label{sec:synthesis}
While we have discussed how advice provided via a TRS can facilitate the learning of an unknown tree language, this section addresses the converse: \emph{synthesis} of rewrite rules consistent with a given regular tree language. 
Given a tree automaton recognizing a tree language, the goal is to find a TRS consistent with that language.
A synthesized TRS can serve as an explanation of the regular tree language, capturing its key properties through concise rewrite rules.
We consider the automaton to be given rather than learned, as we show that even this easier variant of the problem is already difficult.

We begin by observing that the transition function of a DFTA $\aut$  can be viewed as a ground term rewriting system. 
The resulting TRS $\rew_{\aut}$ \emph{completely characterizes} the target language, i.e.,
if $\rew_{\aut}$ is used as the advice TRS, the language of $\aut$ can be learned with a single equivalence query. 

\begin{fact}
Given a DFTA $\aut$ we can compute a ground TRS $\rew_{\aut}$ that completely characterizes $\lang(\aut)$.
\end{fact}

While $\rew_{\aut}$ demonstrates that the number of equivalence queries can be reduced to one, 
it is not practically relevant, as it is derived directly from the automaton being learned.
Such detailed advice essentially describes the automaton itself, rendering the learning process redundant. 
Furthermore, this TRS is not structurally insightful as it does not explain 
the tree language any better than the automaton itself.
Therefore, we are interested in  synthesizing \textbf{non-ground} rewrite rules, where
the left-hand side is a non-ground term, that are consistent with a given tree language and possibly small. 

Formally, we study the following problem:
\begin{definition}[Synthesis of TRS] 
Given a DFTA $\aut$ over a signature $\Sig$, the Rule Synthesis problem is to find a non-ground rewrite rule 
$l \to r$, where $l,r \in \trees(\Sig, \Vars)$, 
such that (1)~$l \to r$ is consistent with $\lang(\aut)$,
(2)~the value $|l|+|r|$ is minimal among rules satisfying (1).
\end{definition}

However, the hardness of the Rule Synthesis problem emerges already in the case of unary and constant symbols.

\paragraph{The unary case.}
Consider a signature $\SigUnary$ consisting of unary symbols $a,b,c$ and a constant $\#$. 
Let $\aut$ be a DFTA over $\SigUnary$.
Observe that a ground term $t$ over $\SigUnary$ corresponds to a state of $\aut$ with 
$\deltaExt{\aut}(t) = \deltaFinal(t)$, while a non-ground term $s$ corresponds to a unary state transformation, i.e., 
$\deltaExt{\aut}(s)$ is a function from $Q$ to $Q$.

The rule $l \to r$ is consistent with $\lang(\aut)$ if and only if 
the transformation induced by $l$ is a constant function equal to $\deltaFinal(r)$.
If we view the unary symbols in $l$ (read from bottom to top) as a word $w_l$ over $\set{a,b,c}$, 
then $w_l$ is a \textbf{synchronizing word} in $\aut$ regarded as a DFA. 

While it can be decided in polynomial time whether a DFA has a synchronizing word~\cite{cerny1964poznamka}, 
finding the shortest synchronizing words in DFA is \NP-complete~\cite{Eppstein90} and 
it is even \NP-hard to approximate it within a constant factor~\cite{Berlinkov14}. 
The hardness of approximation carries directly from DFA to DFTA.

\begin{lemma}
If $\PTime \neq \NP$, there is no polynomial-time algorithm that, given a DFTA $\aut$ over $\SigUnary$, approximates within a constant ratio 
the size of the shortest rewrite rule $l \to r$ such that $r$ is a ground term and $l$ is a non-ground term, and 
$l \to r$ is consistent with $\lang(\aut)$. 
\end{lemma}	
\begin{proof}
Consider a DFA $\autB$. 
Let $\aut$ be the associated DFTA. 
Consider a rewrite rule $l \to r$ such that $r$ is a ground term and $l$ is a non-ground term.
If $l \to r$ is consistent with $\lang(\aut)$, then $\deltaExt{\aut}(l)$ is a constant function.
The symbols along $l$ form a synchronizing word $w_S$ in $\autB$, and hence $|w_S| \leq |l|+|r|$.
Conversely, if $w_S$ is a synchronizing word in $\autB$, then the context $c_S$ corresponding to $w_S$ satisfies
$|c_S| = |w_S|+1$. Now, $c_S(X) \to c_S(\#)$ is a rewrite rule that is 
(1)~consistent with $\lang(\aut)$, 
(2)~$c_S(X)$ is non-ground and $c_S(\#)$ is ground, and
(3)~$|c_S(X)| + |c_S(\#)| = 2|w_S|+2$. 
It follows that for the shortest rewrite rule $l \to r$ consistent with $\lang(\aut)$ and the shortest synchronizing word $w_S$ 
for $\autB$ we have $|w_S| \leq |l| + |r| \leq 2|w_S|+2$. 
Finally, since approximation of the shortest synchronizing word within any constant factor is impossible unless $\PTime = \NP$,
the result follows.
\end{proof}

On the positive side, if we drop the minimality condition and restrict to contexts, 
the existence of such a rule is decidable in polynomial time similarly to~\cite{cerny1964poznamka}:
\begin{lemma}
It is decidable in polynomial time, given a DFTA $\aut$ over $\Sig$, whether there exists a rule $l \to r$ such that
(1)~$l$ is a context and $r$ is a ground term, and
(2)~$l \to r$ is consistent with $\lang(\aut)$. 
\end{lemma} 
\begin{proof}
We show this by the reduction to finding a synchronizing word for DFA. 
Let $Q$ be the set of states of $\aut$ and $n = |Q|$. 
We define $\Sigma$ obtained from $\Sig$ as follows: for any symbol $f \in \Sig$ of arity $k>1$ 
we generate $k \cdot n^{k-1}$ letters: for every position $i$ in $f$ and 
 every sequence of states $\vec{q} \in Q^{k-1}$ we define a letter $a[f,i,\vec{q}]$.

Now, we define a DFA $\autB$ such that its set of states is $Q$, $q_0$ is any state and the set of accepting states $F$ is empty 
($q_0$ and $F$ are irrelevant for synchronization). 
The transition relation is defined as follows: for $q \in Q$ and $a[f,i,\vec{q}] \in \Sigma$ we have
$\delta_{\autB}(p,a[f,i,\vec{q}]) = \delta_{\aut}(\vec{q}',f)$, where
$\vec{q}'$ is obtained from $\vec{q}$ by inserting state $p$ at the position $i$ and shifting all subsequent states, i.e., 
$a[f,i,\vec{q}]$ behaves as a context with the hole at position $i$ and states from $\vec{q}$ at the remaining positions.

Observe that words over $\Sigma$ correspond to reduced contexts, i.e., context where ground terms are pruned to the shortest trees.
If there is a rule $l \to r$ is consistent with $\lang(\aut)$, where $l$ is a context and $r$ is a ground term, then 
$l$ can be transformed to the corresponding synchronizing word by traversing the path from the root to the hole and transforming each context
to  the corresponding letter. 
Conversely, any synchronizing word can be transformed back into a context $c$ such that $\deltaExt{\aut}(c)$ is a constant function 
returning state $q_c$ for every argument. Let $t_c$ be any tree to which $\aut$ assigns the state $q_c$ satisfies (1) and (2). 
\end{proof}

Not only is the complexity of rule synthesis challenging, but there is also no single clear objective.

\paragraph{Synthesis objectives.}
While non-ground rules $l \to r$ are more meaningful than ground rules, 
several questions remain regarding the objectives for synthesis:
\begin{itemize}
    \item Should both terms $l$ and $r$ be non-ground, or only $l$?
    \item Should $l$ be permitted to contain constants, or only variables?
    \item Should $l$ be restricted to linear terms, or allowed to contain repeated variables (non-linear terms)?
\end{itemize}
Investigating these variations is essential for the design of future synthesis algorithms. 
Consequently, we leave the comprehensive synthesis of rewrite rules as an open problem for future work.
 
\section{Conclusions}
We have presented a method that leverages structural knowledge of a target tree language to significantly reduce the number of equivalence queries required during active learning. 
In contrast to the word-automata setting, the branching structure of trees means that checking consistency of 
a TRS with a tree automaton, which is the backbone of our inference algorithm, entails a higher computational complexity.

Thus, a key direction for our future work involves developing further heuristics and investigating specific subclasses of TRSs 
to reduce the computational cost of consistency checking. While we have established several complexity bounds, the complexity of checking consistency of linear rules remains an open problem.

Beyond theoretical refinements, we aim to explore broader applications of this framework. 
Building on successful experiments with the associativity rewrite rule on synthetic data, 
we believe the framework is well-suited for more complex domains. 
Specifically, since our advice mechanism can express various properties of Description Logics (DL), applying active tree automaton learning with advice to learn DL terminologies represents a promising research path.

Finally, while we have only briefly explored the synthesis of consistent TRSs, we consider this a promising frontier for future research. 
Identifying the right objectives for such synthesis will be a prerequisite for developing robust automated advice-generation tools.

\section*{Acknowledgments}
This work was supported by the National Science Centre (NCN), Poland under grant 2024/53/B/ST6/01620.
We thank the anonymous reviewers at ICJAI 2026 for their valuable feedback.
\bibliographystyle{plain}
\bibliography{papers}

@inproceedings{ecai25,
  author       = {Michal Fica and
                  Jan Otop},
  title        = {Active Automata Learning with Advice},
  booktitle    = {{ECAI} 2025},
  series       = {Frontiers in Artificial Intelligence and Applications},
  volume       = {413},
  pages        = {1655--1662},
  publisher    = {{IOS} Press},
  year         = {2025}  
  
}

@inproceedings{learningEL,
  author       = {Matteo Magnini and Riccardo Squarcialupi and Martin Tunge Sterri and Ana Ozaki and Rivera Castillo},
  title        = {Actively Learning EL Terminologies from Large Language Models},
  booktitle    = {{ECAI} 2025},
  series       = {Frontiers in Artificial Intelligence and Applications},
  volume       = {413},
  pages        = {1792--1799},
  publisher    = {{IOS} Press},
  year         = {2025}  
}

@article{Berlinkov14,
  author       = {Mikhail V. Berlinkov},
  title        = {Approximating the Minimum Length of Synchronizing Words Is Hard},
  journal      = {Theory Comput. Syst.},
  volume       = {54},
  number       = {2},
  pages        = {211--223},
  year         = {2014},
  url          = {https://doi.org/10.1007/s00224-013-9511-y},
  doi          = {10.1007/S00224-013-9511-Y},
  bibsource    = {dblp computer science bibliography, https://dblp.org}
}

@article{Eppstein90,
  author       = {David Eppstein},
  title        = {Reset Sequences for Monotonic Automata},
  journal      = {{SIAM} J. Comput.},
  volume       = {19},
  number       = {3},
  pages        = {500--510},
  year         = {1990},
  url          = {https://doi.org/10.1137/0219033},
  doi          = {10.1137/0219033},
  bibsource    = {dblp computer science bibliography, https://dblp.org}
}

@article{transitiveReduction,
author = {Aho, A. V. and Garey, M. R. and Ullman, J. D.},
title = {The Transitive Reduction of a Directed Graph},
journal = {SIAM Journal on Computing},
volume = {1},
number = {2},
pages = {131-137},
year = {1972},
doi = {10.1137/0201008},
URL = {https://doi.org/10.1137/0201008},
eprint = {https://doi.org/10.1137/0201008}
}

@article{cerny1964poznamka,
  title={Pozn{\'a}mka k homog{\'e}nnym experimentom s kone{\v{c}}n{\`y}mi automatmi},
  author={{\v{C}}ern{\`y}, J{\'a}n},
  journal={Matematicko-fyzik{\'a}lny {\v{c}}asopis},
  volume={14},
  number={3},
  pages={208--216},
  year={1964},
  publisher={Mathematical Institute of the Slovak Academy of Sciences}
}

@book{tata,
  TITLE = {{Tree Automata Techniques and Applications}},
  AUTHOR = {Comon, Hubert and Dauchet, Max and Gilleron, R{\'e}mi and Jacquemard, Florent and Lugiez, Denis and L{\"o}ding, Christof and Tison, Sophie and Tommasi, Marc},
  URL = {https://inria.hal.science/hal-03367725},
  PAGES = {262},
  YEAR = {2008},
  HAL_ID = {hal-03367725},
  HAL_VERSION = {v1},
}

@book{BaaderBook,
  author       = {Franz Baader and
                  Tobias Nipkow},
  title        = {Term rewriting and all that},
  publisher    = {Cambridge University Press},
  year         = {1998},
  isbn         = {978-0-521-45520-6},
  bibsource    = {dblp computer science bibliography, https://dblp.org}
}

@inproceedings{DBLP:conf/ijcai/KosalaBBB03,
  author       = {Raymond Kosala and
                  Maurice Bruynooghe and
                  Jan Van den Bussche and
                  Hendrik Blockeel},
  title        = {Information Extraction from Web Documents Based on Local Unranked
                  Tree Automaton Inference},
  booktitle    = {IJCAI 2003},
  pages        = {403--408},
  publisher    = {Morgan Kaufmann},
  year         = {2003},
  url          = {http://ijcai.org/Proceedings/03/Papers/060.pdf},
  bibsource    = {dblp computer science bibliography, https://dblp.org}
}

@article{DLBook,
  title={An introduction to description logics.},
  author={Nardi, Daniele and Brachman, Ronald J and others},
  journal={Description logic handbook},
  volume={1},
  pages={40},
  year={2003}
}

@inproceedings{DBLP:conf/rta/BonevaT05,
  author       = {Iovka Boneva and
                  Jean{-}Marc Talbot},
  editor       = {J{\"{u}}rgen Giesl},
  title        = {Automata and Logics for Unranked and Unordered Trees},
  booktitle    = {{RTA} 2005},
  series       = {Lecture Notes in Computer Science},
  volume       = {3467},
  pages        = {500--515},
  publisher    = {Springer},
  year         = {2005},
  url          = {https://doi.org/10.1007/978-3-540-32033-3\_36},
  doi          = {10.1007/978-3-540-32033-3\_36},
  bibsource    = {dblp computer science bibliography, https://dblp.org}
}

@inproceedings{DBLP:conf/dlt/DrewesH03,
  author       = {Frank Drewes and
                  Johanna H{\"{o}}gberg},
  editor       = {Zolt{\'{a}}n {\'{E}}sik and
                  Zolt{\'{a}}n F{\"{u}}l{\"{o}}p},
  title        = {Learning a Regular Tree Language from a Teacher},
  booktitle    = {{DLT} 2003},
  series       = {Lecture Notes in Computer Science},
  volume       = {2710},
  pages        = {279--291},
  publisher    = {Springer},
  year         = {2003},
  url          = {https://doi.org/10.1007/3-540-45007-6\_22},
  doi          = {10.1007/3-540-45007-6\_22},
  bibsource    = {dblp computer science bibliography, https://dblp.org}
}

@inproceedings{DBLP:conf/icgi/HabrardO06,
  author       = {Amaury Habrard and
                  Jos{\'{e}} Oncina},
  editor       = {Yasubumi Sakakibara and
                  Satoshi Kobayashi and
                  Kengo Sato and
                  Tetsuro Nishino and
                  Etsuji Tomita},
  title        = {Learning Multiplicity Tree Automata},
  booktitle    = {{ICGI} 2006},
  series       = {Lecture Notes in Computer Science},
  volume       = {4201},
  pages        = {268--280},
  publisher    = {Springer},
  year         = {2006},
  url          = {https://doi.org/10.1007/11872436\_22},
  doi          = {10.1007/11872436\_22},
  bibsource    = {dblp computer science bibliography, https://dblp.org}
}

@article{charatonik1999automata,
  title={Automata on DAG representations of finite trees},
  author={Charatonik, Witold},
  year={1999},
  publisher={Max-Planck-Institut f{\"u}r Informatik}
}

@article{DBLP:journals/ipl/AnantharamanNR05,
  author       = {Siva Anantharaman and
                  Paliath Narendran and
                  Micha{\"{e}}l Rusinowitch},
  title        = {Closure properties and decision problems of dag automata},
  journal      = {Inf. Process. Lett.},
  volume       = {94},
  number       = {5},
  pages        = {231--240},
  year         = {2005},
  url          = {https://doi.org/10.1016/j.ipl.2005.02.004},
  doi          = {10.1016/J.IPL.2005.02.004},
  bibsource    = {dblp computer science bibliography, https://dblp.org}
}

@inproceedings{bhattamishraautomata,
title={Automata Learning and Identification of the Support of Language Models},
author={Satwik Bhattamishra and Michael Hahn and Varun Kanade},
booktitle={The Fourteenth International Conference on Learning Representations},
year={2026},
url={https://openreview.net/forum?id=L8SMNWsxfK}
}

@misc{vazquezchanlatte2025ll,
      title={$L^*LM$: Learning Automata from Examples using Natural Language Oracles}, 
      author={Marcell Vazquez-Chanlatte and Karim Elmaaroufi and Stefan J. Witwicki and Matei Zaharia and Sanjit A. Seshia},
      year={2025},
      eprint={2402.07051},
      archivePrefix={arXiv},
      primaryClass={cs.LG},
      url={https://arxiv.org/abs/2402.07051}, 
}

@inproceedings{DBLP:conf/sle/Raselimo021,
  author       = {Moeketsi Raselimo and
                  Bernd Fischer},
  editor       = {Eelco Visser and
                  Dimitris S. Kolovos and
                  Emma S{\"{o}}derberg},
  title        = {Automatic grammar repair},
  booktitle    = {{SLE} 2021},
  pages        = {126--142},
  publisher    = {{ACM}},
  year         = {2021},
  url          = {https://doi.org/10.1145/3486608.3486910},
  doi          = {10.1145/3486608.3486910},
  bibsource    = {dblp computer science bibliography, https://dblp.org}
}

@inproceedings{DBLP:conf/icdt/GrienenbergerR19,
  author       = {{\'{E}}milie Grienenberger and
                  Martin Ritzert},
  editor       = {Pablo Barcel{\'{o}} and
                  Marco Calautti},
  title        = {Learning Definable Hypotheses on Trees},
  booktitle    = {{ICDT} 2019},
  series       = {LIPIcs},
  volume       = {127},
  pages        = {24:1--24:18},
  publisher    = {Schloss Dagstuhl - Leibniz-Zentrum f{\"{u}}r Informatik},
  year         = {2019},
  url          = {https://doi.org/10.4230/LIPIcs.ICDT.2019.24},
  doi          = {10.4230/LIPICS.ICDT.2019.24},
  bibsource    = {dblp computer science bibliography, https://dblp.org}
}

@inproceedings{DBLP:conf/pldi/YaghmazadehKDC16,
  author       = {Navid Yaghmazadeh and
                  Christian Klinger and
                  Isil Dillig and
                  Swarat Chaudhuri},
  editor       = {Chandra Krintz and
                  Emery D. Berger},
  title        = {Synthesizing transformations on hierarchically structured data},
  booktitle    = {{PLDI} 2016},
  pages        = {508--521},
  publisher    = {{ACM}},
  year         = {2016},
  url          = {https://doi.org/10.1145/2908080.2908088},
  doi          = {10.1145/2908080.2908088},
  bibsource    = {dblp computer science bibliography, https://dblp.org}
}

@article{angluin1987learning,
  title={Learning regular sets from queries and counterexamples},
  author={Angluin, Dana},
  journal={Information and computation},
  volume={75},
  number={2},
  pages={87--106},
  year={1987},
  publisher={Elsevier}
}

@inproceedings{TTTalgorithm,
  author       = {Malte Isberner and
                  Falk Howar and
                  Bernhard Steffen},
  title        = {The {TTT} Algorithm: {A} Redundancy-Free Approach to Active Automata
                  Learning},
  booktitle    = {{RV} 2014},
  pages        = {307--322},
  year         = {2014},
  url          = {https://doi.org/10.1007/978-3-319-11164-3\_26},
  doi          = {10.1007/978-3-319-11164-3\_26},
  bibsource    = {dblp computer science bibliography, https://dblp.org}
}

@article{ADTalgorithm,
  author       = {Markus Theo Frohme},
  title        = {Active Automata Learning with Adaptive Distinguishing Sequences},
  journal      = {CoRR},
  volume       = {abs/1902.01139},
  year         = {2019},
  url          = {http://arxiv.org/abs/1902.01139},
  eprinttype    = {arXiv},
  eprint       = {1902.01139},
  bibsource    = {dblp computer science bibliography, https://dblp.org}
}

@inproceedings{Lsharp,
  author       = {Frits W. Vaandrager and
                  Bharat Garhewal and
                  Jurriaan Rot and
                  Thorsten Wi{\ss}mann},
  title        = {A New Approach for Active Automata Learning Based on Apartness},
  booktitle    = {{TACAS} 2022},
  pages        = {223--243},
  year         = {2022},
  url          = {https://doi.org/10.1007/978-3-030-99524-9\_12},
  doi          = {10.1007/978-3-030-99524-9\_12},
  bibsource    = {dblp computer science bibliography, https://dblp.org}
}

@inproceedings{KrugerJR24,
  author       = {Loes Kruger and
                  Sebastian Junges and
                  Jurriaan Rot},
  title        = {Small Test Suites for Active Automata Learning},
  booktitle    = {{TACAS} 2024},
  pages        = {109--129},
  year         = {2024},
  url          = {https://doi.org/10.1007/978-3-031-57249-4\_6},
  doi          = {10.1007/978-3-031-57249-4\_6},
  bibsource    = {dblp computer science bibliography, https://dblp.org}
}

@article{MarusicW15,
  author    = {Ines Marusic and
               James Worrell},
  title     = {Complexity of equivalence and learning for multiplicity tree automata},
  journal   = {Journal of Machine Learning Research},
  volume    = {16},
  pages     = {2465--2500},
  year      = {2015},
}

@inproceedings{DierlFHJST24,
  author       = {Simon Dierl and
                  Paul Fiterau{-}Brostean and
                  Falk Howar and
                  Bengt Jonsson and
                  Konstantinos Sagonas and
                  Fredrik T{\aa}quist},
  title        = {Scalable Tree-based Register Automata Learning},
  booktitle    = {{TACAS} 2024},
  pages        = {87--108},
  year         = {2024},
  url          = {https://doi.org/10.1007/978-3-031-57249-4\_5},
  doi          = {10.1007/978-3-031-57249-4\_5},
  bibsource    = {dblp computer science bibliography, https://dblp.org}
}

@article{ModelLearning,
  title={Model learning},
  author={Vaandrager, Frits},
  journal={Communications of the ACM},
  volume={60},
  number={2},
  pages={86--95},
  year={2017},
  publisher={ACM New York, NY, USA}
}

@inproceedings{BaaderTrees,
  author       = {Franz Baader and
                  Paliath Narendran},
  editor       = {Henri Prade},
  title        = {Unification of Concept Terms in Description Logics},
  booktitle    = {ECAI 1998},
  pages        = {331--335},
  publisher    = {John Wiley and Sons},
  year         = {1998},
  bibsource    = {dblp computer science bibliography, https://dblp.org}
}

@article{DBLP:journals/mst/DrewesH07,
  author       = {Frank Drewes and
                  Johanna H{\"{o}}gberg},
  title        = {Query Learning of Regular Tree Languages: How to Avoid Dead States},
  journal      = {Theory Comput. Syst.},
  volume       = {40},
  number       = {2},
  pages        = {163--185},
  year         = {2007},
  url          = {https://doi.org/10.1007/s00224-005-1233-3},
  doi          = {10.1007/S00224-005-1233-3},
  bibsource    = {dblp computer science bibliography, https://dblp.org}
}

@inproceedings{barbot2021extracting,
  title={Extracting context-free grammars from recurrent neural networks using tree-automata learning and a* search},
  author={Barbot, Beno{\^\i}t and Bollig, Benedikt and Finkel, Alain and Haddad, Serge and Khmelnitsky, Igor and Leucker, Martin and Neider, Daniel and Roy, Rajarshi and Ye, Lina},
  booktitle={{ICGI} 2021},
  pages={113--129},
  year={2021},
  organization={PMLR}
}

@inproceedings{DBLP:conf/fossacs/HeerdtKR021,
  author       = {Gerco van Heerdt and
                  Tobias Kapp{\'{e}} and
                  Jurriaan Rot and
                  Alexandra Silva},
  editor       = {Stefan Kiefer and
                  Christine Tasson},
  title        = {Learning Pomset Automata},
  booktitle    = {{ETAPS} 2021},
  volume       = {12650},
  pages        = {510--530},
  publisher    = {Springer},
  year         = {2021},
  url          = {https://doi.org/10.1007/978-3-030-71995-1\_26},
  doi          = {10.1007/978-3-030-71995-1\_26},
  bibsource    = {dblp computer science bibliography, https://dblp.org}
}

@article{DBLP:journals/jalc/TirnaucaT07,
  author       = {Catalin Ionut Tirn\u{a}uc\u{a} and
                  Cristina Tirn\u{a}uc\u{a}},
  title        = {Learning Regular Tree Languages from Correction and Equivalence Queries},
  journal      = {J. Autom. Lang. Comb.},
  volume       = {12},
  number       = {4},
  pages        = {501--524},
  year         = {2007},
  url          = {https://doi.org/10.25596/jalc-2007-501},
  doi          = {10.25596/JALC-2007-501},
  bibsource    = {dblp computer science bibliography, https://dblp.org}
}

@article{DBLP:journals/tcs/Kasprzik13,
  author       = {Anna Kasprzik},
  title        = {Four one-shot learners for regular tree languages and their polynomial
                  characterizability},
  journal      = {Theor. Comput. Sci.},
  volume       = {485},
  pages        = {85--106},
  year         = {2013},
  url          = {https://doi.org/10.1016/j.tcs.2013.01.003},
  doi          = {10.1016/J.TCS.2013.01.003},
  bibsource    = {dblp computer science bibliography, https://dblp.org}
}

@inproceedings{DBLP:conf/aaai/NitayFZ21,
  author       = {Dolav Nitay and
                  Dana Fisman and
                  Michal Ziv{-}Ukelson},
  title        = {Learning of Structurally Unambiguous Probabilistic Grammars},
  booktitle    = {{AAAI} 2021},
  pages        = {9170--9178},
  publisher    = {{AAAI} Press},
  year         = {2021},
  url          = {https://doi.org/10.1609/aaai.v35i10.17107},
  doi          = {10.1609/AAAI.V35I10.17107},
  bibsource    = {dblp computer science bibliography, https://dblp.org}
}

@inproceedings{BjorklundBE17,
  author       = {Henrik Bj{\"{o}}rklund and
                  Johanna Bj{\"{o}}rklund and
                  Petter Ericson},
  editor       = {Arnaud Carayol and
                  Cyril Nicaud},
  title        = {On the Regularity and Learnability of Ordered {DAG} Languages},
  booktitle    = {{CIAA} 2017},
  series       = {Lecture Notes in Computer Science},
  volume       = {10329},
  pages        = {27--39},
  publisher    = {Springer},
  year         = {2017},
  url          = {https://doi.org/10.1007/978-3-319-60134-2\_3},
  doi          = {10.1007/978-3-319-60134-2\_3},
  bibsource    = {dblp computer science bibliography, https://dblp.org}
}

@misc{experimentsRepo,
  author = {Jakub Kopystiański and Jan Otop},
  title = {Learning Tree Automata with Term Rewriting: code and data},
  year = {2026},
  publisher = {GitHub},
  journal = {GitHub repository},
  howpublished = {\url{https://github.com/jotop/LearnDFTAwithTRS}},
  commit = {}
}

@inproceedings{ijcai26,
  author={Jakub Kopystiański and Jan Otop},
  title={Learning Tree Automata with Term Rewriting}, 
  booktitle    = {to appear at {IJCAI} 2026},
  publisher    = {ijcai.org},
  year         = {2026},
}

@article{DBLP:journals/dam/Tovey84,
  author       = {Craig A. Tovey},
  title        = {A simplified NP-complete satisfiability problem},
  journal      = {Discret. Appl. Math.},
  volume       = {8},
  number       = {1},
  pages        = {85--89},
  year         = {1984},
  url          = {https://doi.org/10.1016/0166-218X(84)90081-7},
  doi          = {10.1016/0166-218X(84)90081-7},
  bibsource    = {dblp computer science bibliography, https://dblp.org}
}

\end{document}